\documentclass[a4paper]{article}
\usepackage[margin=1in]{geometry}
\usepackage{microtype}
\usepackage{graphicx}

\usepackage[algoruled,linesnumbered,algo2e,vlined]{algorithm2e}

\usepackage{amsmath,amssymb}
\usepackage{url}
\usepackage{rotating}
\usepackage{comment}
\usepackage{multirow}
\usepackage{amsthm}
\theoremstyle{plain}
\newtheorem{theorem}{Theorem}
\newtheorem{definition}[theorem]{Definition}
\newtheorem{lemma}[theorem]{Lemma}
\newtheorem{corollary}[theorem]{Corollary}

\newtheorem{fact}[theorem]{Fact}
\newtheorem{example}[theorem]{Example}

\newcommand{\LCE}{\mathsf{LCE}}

\newcommand{\freq}{\mathsf{Freq}}

\newcommand{\extract}{\mathsf{Extract}}

\newcommand{\emptystr}{\varepsilon} 

\newcommand{\slp}{\mathcal{S}} 
\newcommand{\vars}{\mathcal{V}}
\newcommand{\rules}{\mathcal{D}}

\newcommand{\hrules}[1]{\mathcal{D}_{#1}}

\newcommand{\dtree}{\mathcal{T}} 
\newcommand{\gtext}[1]{\mathit{T}_{#1}} 
\newcommand{\letters}[1]{\Sigma_{#1}} 
\newcommand{\nletters}[1]{\hat{\letters{}}_{#1}}
\newcommand{\bletters}[1]{\bchar{\letters{}}_{#1}}
\newcommand{\lletters}[1]{\lchar{\letters{}}_{#1}}
\newcommand{\rletters}[1]{\rchar{\letters{}}_{#1}}

\newcommand{\lchar}[1]{\acute{#1}}
\newcommand{\rchar}[1]{\grave{#1}}
\newcommand{\bchar}[1]{\ddot{#1}}

\newcommand{\lml}[1]{\mathsf{LML}({#1})} 
\newcommand{\rml}[1]{\mathsf{RML}({#1})} 

\newcommand{\pol}{\mathsf{PopOutLet}} 
\newcommand{\pil}{\mathsf{PopInLet}} 

\newcommand{\vocc}[1]{\mathsf{VOcc}({#1})} 

\newcommand{\PSeq}{\mathit{PSeq}}

\newcommand{\rg}[1]{\mathcal{S}_{#1}} 

\newcommand{\val}[1]{\mathit{val}_{#1}} 

\newcommand{\bcomp}{\mathsf{BComp}}
\newcommand{\pcomp}{\mathsf{PComp}}
\newcommand{\recomp}{\mathsf{TtoG}}
\newcommand{\simrecomp}{\mathsf{SimTtoG}}
\newcommand{\gtog}{\mathsf{GtoG}}

\title{
  Longest Common Extensions with Recompression
}

\author{
  Tomohiro~I\\
  {Kyushu Institute of Technology, Japan}\\
  {\texttt{tomohiro@ai.kyutech.ac.jp}}\\
}

\date{}
\begin{document}
\maketitle

\begin{abstract}
  Given two positions $i$ and $j$ in a string $\gtext{}$ of length $N$, 
  a \emph{longest common extension (LCE) query} asks for the length of the longest common prefix
  between suffixes beginning at $i$ and $j$.
  A compressed LCE data structure is a data structure that stores $\gtext{}$ in a compressed form while supporting fast LCE queries.
  In this article we show that the \emph{recompression} technique is a powerful tool for compressed LCE data structures.
  We present a new compressed LCE data structure of size $O(z \lg (N/z))$ that supports LCE queries in $O(\lg N)$ time, 
  where $z$ is the size of Lempel-Ziv 77 factorization without self-reference of $\gtext{}$.
  Given $\gtext{}$ as an uncompressed form,
  we show how to build our data structure in $O(N)$ time and space.
  Given $\gtext{}$ as a grammar compressed form, i.e., an straight-line program of size $n$ generating $\gtext{}$,
  we show how to build our data structure in $O(n \lg (N/n))$ time and $O(n + z \lg (N/z))$ space.
  Our algorithms are deterministic and always return correct answers.
\end{abstract}

\section{Introduction}
Given two positions $i$ and $j$ in a text $\gtext{}$ of length $N$, 
a \emph{longest common extension (LCE) query} $\LCE(i,j)$ asks for the length of the longest common prefix
between suffixes beginning at $i$ and $j$.
Since LCE queries play a central role in many string processing algorithms~(see text book~\cite{Gusfield97} for example),
efficient LCE data structures have been extensively studied.
If we are allowed to use $O(N)$ space, optimal $O(1)$ query time can be achieved by, e.g., 
lowest common ancestor queries~\cite{Bender2005Lca} on the suffix tree of $\gtext{}$.
However, $O(N)$ space can be too expensive nowadays as the size of strings to be processed becomes quite large.
Thus, recent studies focus on more space efficient solutions.

Roughly there are three scenarios:
Several authors have studied tradeoffs among query time, 
construction time and data structure size~\cite{PT08,bille14:_time,bille15:_longes_common_exten_sublin_space,Tanimura2016LCE_EfficitnConstruction};
In~\cite{Prezza2016InplaceLCE}, Prezza presented in-place LCE data structures 
showing that the memory space for storing $\gtext{}$ can be replaced with an LCE data structure while retaining optimal substring extraction time;
LCE data structures working on grammar compressed representation of $\gtext{}$ were studied in~\cite{I2015Drg,bille13:_finger_compr_strin,BilleCCG15,Nishimoto2016DynamicLCE_CompressedSpace}.

In this article we pursue the third scenario, which is advantageous when $\gtext{}$ is highly compressible.
In grammar compression, $\gtext{}$ is represented by a Context Free Grammar (CFG) that generates $\gtext{}$ and only $\gtext{}$.
In particular CFGs in Chomsky normal form, called Straight Line Programs (SLPs), are often considered
as any CFG can be easily transformed into an SLP without changing the order of grammar size.
Let $\slp$ be an arbitrary SLP of size $n$ generating $\gtext{}$.
Bille et al.~\cite{BilleCCG15} showed a Monte Carlo randomized data structure of $O(n)$ space
that supports LCE queries in $O(\lg N + \lg^2 \ell)$ time, where $\ell$ is the answer to the LCE query.
Because their algorithm is based on Karp-Rabin fingerprints, the answer is correct w.h.p (with high probability).
If we always expect correct answers, we have to verify fingerprints in preprocessing phase, 
spending either $O(N \lg N)$ time (w.h.p.) and $O(N)$ space or $O(\frac{N^2}{n} \lg N)$ time (w.h.p.) and $O(n)$ space.

For a deterministic solution, 
I et al.~\cite{I2015Drg} proposed an $O(n^2)$-space data structure, 
which can be built in $O(n^2h)$ time and $O(n^2)$ space from $\slp$, and
supports LCE queries in $O(h \lg N)$ time, where $h$ is the height of $\slp$.
As will be stated in Theorem~\ref{theo:lce_from_slp}, we outstrip this result.

Our work is most similar to that presented in~\cite{Nishimoto2016DynamicLCE_CompressedSpace}.
They showed that the signature encoding~\cite{DBLP:journals/algorithmica/MehlhornSU97} of $\gtext{}$, 
a special kind of CFGs that can be stored in $O(z \lg N \lg^* N)$ space,
can support LCE queries in $O(\lg N + \lg \ell \lg^* N)$ time,
where $z$ is the size of LZ77 factorization\footnote{Note that there are several variants of LZ77 factorization. In this article we refer to the one that is known as the \emph{f-factorization without self-reference} as LZ77 factorization unless otherwise noted.} of $\gtext{}$ and $\lg^*$ is the iterated logarithm.
The signature encoding is based on localy consistent parsing technique, 
which determines the parsing of a string by local surrounding.
A key property of the signature encoding is that 
any occurrence of the same substring of length $\ell$ in $\gtext{}$ is guaranteed to be compressed in almost same way 
leaving only $O(\lg \ell \lg^* N)$ discrepancies in its surrounding.
As a result, an LCE query can be answered by tracing the $O(\lg \ell \lg^* N)$ surroundings
created over two occurrences of the longest common extension.
The algorithm is quite simple as we simply simulate the traversal of the derivation tree on the CFG
while matching substrings by appearances of the common variables, which takes $O(\lg N + \lg \ell \lg^* N)$ time.
Note that the cost $O(\lg N)$ is needed anyway to traverse the derivation tree of height $O(\lg N)$ from the root.

In this article we show that CFGs created by the \emph{recompression} technique exhibit a similar property
that can be used to answer LCE queries in $O(\lg N)$ time.
In recent years recompression has been proved to be a powerful tool in problems related to grammar compression~\cite{Jez2015Aog,Jez2015FFC,Jez2014Aos}
and word equations~\cite{Jez2016OneVariableWordEquation_LinearTime,Jez2016Recompression_WordEquations}.
The main component of recompression is to replace some pairs in a string with variables of the CFG.
Although we use global information (like the frequencies of pairs in the string) to determine which pairs to be replaced, 
the pairing itself is done very locally, i.e., ``all'' occurrences of the pairs are replaced.
Then we can show that any occurrence of the same substring in $\gtext{}$ is guaranteed to be compressed in almost same way
leaving only $O(\lg N)$ discrepancies in its surrounding.
This leads to an $O(\lg N)$-time algorithm to answer LCE queries, 
improving the $O(\lg N + \lg \ell \lg^* N)$-time algorithm of~\cite{Nishimoto2016DynamicLCE_CompressedSpace}.
We also improve the data structure size from $O(z \lg N \lg^* N)$ of~\cite{Nishimoto2016DynamicLCE_CompressedSpace}\footnote{We believe that 
the space complexities of~\cite{Nishimoto2016DynamicLCE_CompressedSpace} can be improved to $O(z \lg (N/z) \lg^* N)$ 
by using the same trick we use in Lemma~\ref{lem:recomp_size}.} to $O(z \lg (N/z))$.

In~\cite{Nishimoto2016DynamicLCE_CompressedSpace}, the authors proposed efficient algorithms to 
build their LCE data structure from various kinds of input as summarized in Table~\ref{table:resultcomparison}.
We achieve a better and cleaner complexity to build our LCE data structure from SLPs.
This has a great impact on compressed string processing, 
in which we are to solve problems on SLPs without decompressing the string explicitly.
For instance, we can apply our result to the problems discussed in Section~7 of~\cite{Nishimoto2016DynamicLCE_CompressedSpace} 
and immediately improve the results (other than Theorem~17).
It should be noted that the data structures in~\cite{Nishimoto2016DynamicLCE_CompressedSpace} also support efficient text edit operations.
We are not sure if our data structures can be efficiently dynamized.

\begin{table}
  \begin{center}
    \label{table:resultcomparison}
    \begin{tabular}{|c|c|c|l|}\hline
      Input      & Construction time              & Construction space         & Reference \\\hline
      $\gtext{}$ & $N f_{\mathcal{A}}$               & $z \lg N \lg^* N$          & Theorem~3 (1a) of~\cite{Nishimoto2016DynamicLCE_CompressedSpace} \\\hline
      $\gtext{}$ & $N$                            & $N$                        & Theorem~3 (1b) of~\cite{Nishimoto2016DynamicLCE_CompressedSpace} \\\hline
      $\slp$     & $n f_{\mathcal{A}} \lg N \lg^* N$ & $n + z \lg N \lg^* N$        & Theorem~3 (3a) of~\cite{Nishimoto2016DynamicLCE_CompressedSpace} \\\hline
      $\slp$     & $n \lg \lg n \lg N \lg^* N$   & $n \lg^* N + z \lg N \lg^* N$ & Theorem~3 (3b) of~\cite{Nishimoto2016DynamicLCE_CompressedSpace} \\\hline
      LZ77       & $z f_{\mathcal{A}} \lg N \lg^* N$ & $z \lg N \lg^* N$            & Theorem~3 (2) of~\cite{Nishimoto2016DynamicLCE_CompressedSpace} \\\hline
      $\gtext{}$ & $N$                           & $N$                         & this work, Theorem~\ref{theo:lce_from_uncompressed} \\\hline
      $\slp$     & $n\lg(N/n)$                   & $n + z \lg (N/z)$           & this work, Theorem~\ref{theo:lce_from_slp} \\\hline
      LZ77       & $z\lg^2(N/z)$                 & $z \lg (N/z)$               & this work, Corollary~\ref{coro:lce_from_lz} \\\hline
    \end{tabular}
    \caption{Comparison of construction time and space between ours and~\cite{Nishimoto2016DynamicLCE_CompressedSpace}, 
             where $N$ is the length of $\gtext{}$, $\slp$ is an SLP of size $n$ generating $\gtext{}$,
             $z$ is the size of LZ77 factorization of $\gtext{}$, and $f_{\mathcal{A}}$ is the time needed for predecessor queries 
             on a set of $z \lg N \lg^* N$ integers from an $N$-element universe.}
  \end{center}
\end{table}

Theorems~\ref{theo:lce_from_uncompressed} and~\ref{theo:lce_from_slp} show our main results.
Note that our data structure is a simple CFG of height $O(\lg N)$
on which we can simulate the traversal of the derivation tree in constant time per move.
Thus, it naturally supports $\extract(i, \ell)$ queries, 
which asks for retrieving the substring $\gtext{}[i..i+\ell-1]$, in $O(\lg N + \ell)$ time.
\begin{theorem}\label{theo:lce_from_uncompressed}
  Given a string $\gtext{}$ of length $N$, we can compute in $O(N)$ time and space a compressed representation of $\gtext{}$ of size $O(z \lg (N/z))$
  that supports $\extract(i, \ell)$ in $O(\lg N + \ell)$ time and $\LCE$ queries in $O(\lg N)$ time.
\end{theorem}

\begin{theorem}\label{theo:lce_from_slp}
  Given an SLP of size $n$ generating a string $\gtext{}$ of length $N$, 
  we can compute in $O(n \lg (N/n))$ time and $O(n + z \lg (N/z))$ space a compressed representation of $\gtext{}$ of size $O(z \lg (N/z))$
  that supports $\extract(i, \ell)$ in $O(\lg N + \ell)$ time and $\LCE$ queries in $O(\lg N)$ time.
\end{theorem}

Suppose that we are given the LZ77-compression of size $z$ of $T$ as an input.
Since we can convert the input into an SLP of size $O(z \lg (N/z))$~\cite{rytter03:_applic_lempel_ziv},
we can apply Theorem~\ref{theo:lce_from_slp} to the SLP and get the next corollary.
\begin{corollary}\label{coro:lce_from_lz}
  Given the LZ77-compression of size $z$ of a string $\gtext{}$ of length $N$, 
  we can compute in $O(z \lg^2 (N/z))$ time and $O(z \lg (N/z))$ space a compressed representation of $\gtext{}$ of size $O(z \lg (N/z))$
  that supports $\extract(i, \ell)$ in $O(\lg N + \ell)$ time and $\LCE$ queries in $O(\lg N)$ time.
\end{corollary}

Technically, this work owes very much to two papers~\cite{Jez2015FFC,Jez2015Aog}.
For instance, our construction algorithm of Theorem~\ref{theo:lce_from_uncompressed}
is essentially the same as the grammar compression algorithm~\cite{Jez2015FFC} based on recompression,
which produces an SLP of size $O(g^* \lg (N/g^*))$ generating an input string $\gtext{}$, where $g^*$ is the smallest grammar size to generate $\gtext{}$.
Our contribution is in discovering the above mentioned property that can be used for fast LCE queries.
Also, we use the property to upper bound the size of our data structure in terms of $z$ rather than $g^*$.
Since it is known that $z \leq g^*$ holds, an upper bound in terms of $z$ is preferable.
The technical issues in our construction algorithm of Theorem~\ref{theo:lce_from_slp} have been tackled in~\cite{Jez2015Aog},
in which the recompression technique is used to solve the fully-compressed pattern matching problems.
However, we make some contributions on top of it:
We give a new observation that simplifies the implementation and analysis of a component of recompression called $\bcomp$ (see Section~\ref{sec:bcomp_on_CFG}).
Also, we achieve a better construction time $O(n \lg (N/n))$ than $O(n \lg N)$ (which is obtained by straightforwardly applying the analysis in~\cite{Jez2015Aog}).

\section{Preliminaries}
An alphabet $\Sigma$ is a set of characters.
A string over $\Sigma$ is an element in $\Sigma^*$.
For any string $w \in \Sigma^{*}$, $|w|$ denotes the length of $w$.
Let $\emptystr$ be the empty string, i.e., $|\emptystr| = 0$.
Let $\Sigma^{+} = \Sigma^{*} \setminus \{ \emptystr \}$.
For any $1 \leq i \leq |w|$, $w[i]$ denotes the $i$-th character of $w$.
For any $1 \leq i \leq j \leq |w|$,
$w[i..j]$ denotes the substring of $w$ beginning at $i$ and ending at $j$.
For convenience, let $w[i..j] = \emptystr$ if $i > j$.
For any $0 \leq i \leq |w|$, $w[1..i]$ (resp.\ $w[|w|-i+1..|w|]$) is called the prefix (resp.\ suffix) of $w$ of length $i$.
We say taht a string $x$ \emph{occurs} at position $i$ in $w$ iff $w[i..i+|x|-1] = x$.
A substring $w[i..j] = c^d~(c \in \Sigma, d \geq 1)$ of $w$ is called a \emph{block} iff it is a maximal run of a single character, 
i.e., $(i = 1 \vee w[i-1] \neq c) \wedge (j = |w| \vee w[j+1] \neq c)$.

The text on which LCE queries are performed is denoted by $\gtext{} \in \Sigma^{*}$ with $N = |\gtext{}|$ throughout this paper.
We assume that $\Sigma$ is an integer alphabet $[1..N^{O(1)}]$ and the standard word RAM model with word size $\Omega(\lg N)$.

The size of our compressed LCE data structure is bounded by $O(z \lg (N/z))$, 
where $z$ is the size of the LZ77 factorization of $\gtext{}$ defined as follows:
\begin{definition}[LZ77 factorization]
The factorization $\gtext{} = f_1 f_2 \cdots f_{z}$ is the LZ77 factorization of $\gtext{}$ iff the following condition holds:
For any $1 \leq i \leq z$, let $p_i = |f_1 f_2 \cdots f_{i-1}| + 1$, then
$f_i = \gtext{}[p_i]$ if $\gtext{}[p_i]$ does not appear in $\gtext{}[1..p_i-1]$,
otherwise $f_i$ is the longest prefix of $\gtext{}[p_i..N]$ that occurs in $\gtext{}[1..p_{i}-1]$.
\end{definition}

\begin{example}
The LZ77 factorization of $\mathtt{abaabaabb}$ is $\mathtt{a} \cdot \mathtt{b} \cdot \mathtt{a} \cdot \mathtt{aba} \cdot \mathtt{ab} \cdot \mathtt{b}$ and $z = 6$.
\end{example}

In this article, we deal with grammar compressed strings,
in which a string is represented by a Context Free Grammar (CFG) generating the string only.
In particular, we consider \emph{Straight-Line Programs (SLPs)} that are CFGs in Chomsky normal form.
Formally, an SLP that generates a string $\gtext{}$ is a triple 
$\slp = (\Sigma, \vars, \rules)$, where $\Sigma$ is the set of characters (terminals),
$\vars$ is the set of variables (non-terminals),
$\rules$ is the set of deterministic production rules whose righthand sides are in $\vars^2 \cup \Sigma$, and
the last variable derives $\gtext{}$.\footnote{We treat the last variable as the starting variable.}
Let $n = |\vars|$.
We treat variables as integers in $[1..n]$ (which should be distinguishable from $\Sigma$ by having extra one bit), 
and $\rules$ as an injective function that maps a variable to its righthand side.
We assume that given any variable $X$ we can access in $O(1)$ time to the data space storing the information of $X$, e.g., $\rules(X)$.
We refer to $n$ as the size of $\slp$ since $\slp$ can be encoded in $O(n)$ space.
Note that $N$ can be as large as $2^{n-1}$, and so, SLPs have a potential to achieve exponential compression.

We extend SLPs by allowing run-length encoded rules whose righthand sides are of the form $X^d$ with $X \in \vars$ and $d \geq 2$,
and call such CFGs \emph{run-length SLPs (RLSLPs)}.
Since a run-length encoded rule can be stored in $O(1)$ space, we still define the size of an RLSLP by the number of variables.

Let us consider the derivation tree $\dtree$ of an RLSLP $\slp$ that generates a string $\gtext{}$,
where we delete all the nodes labeled with terminals for simplicity.
That is, every node in $\dtree$ is labeled with a variable.
The height of $\slp$ is the height of $\dtree$.
We say that a sequence $C = v_1 \cdots v_m$ of nodes is a \emph{chain} iff the nodes are all adjacent in this order, i.e., 
the beginning position of $v_{i+1}$ is the ending position of $v_{i}$ plus one for any $1 \leq i < m$.
$C$ is labeled with the sequence of labels of $v_1 \cdots v_m$.

For any sequence $p \in \vars^*$ of variables,
let $\val{\slp}(p)$ denote the string obtained by concatenating the strings derived from all variables in the sequence.
We omit $\slp$ when it is clear from context.
We say that $p$ generates $\val{}(p)$.
Also, we say that $p$ \emph{occurs} at position $i$ iff there is a chain that is labeled with $p$ and begins at $i$.

The next lemma, which is somewhat standard for SLPs, also holds for RLSLPs.
\begin{lemma}\label{lem:extract}
  For any RSLP $\slp$ of height $h$ generating $\gtext{}$, by storing $|\val{}(X)|$ for every variable $X$,
  we can support $\extract(i, \ell)$ in $O(h + \ell)$ time.
\end{lemma}

\section{LCE data structure built from uncompressed texts}

In this section, we prove Theorem~\ref{theo:lce_from_uncompressed}.
We basically show that the RLSLP obtained by grammar compression algorithm based on recompression~\cite{Jez2015Aog} can be used for fast LCE queries.
In Subsection~\ref{sec:recomp} we first review the recompression and introduce notation we use.
In Subsection~\ref{sec:popseq} we present a new characterization of recompression, which is a key to our contributions.

\subsection{$\recomp$: Grammar compression based on recompression}\label{sec:recomp}
In~\cite{Jez2015Aog} Je\.z proposed an algorithm $\recomp$ to compute an RLSLP of $\gtext{}$ in $O(N)$ time based on 
the recompression technique.\footnote{Indeed, the paper shows how to compute an ``SLP'' of size $O(g^* \lg (N/g^*))$, where $g^*$ is the smallest SLP size to generate $\gtext{}$. In order to estimate the number of SLP's variables needed to represent run-length encoded rules, its analysis becomes much involved.}
Let $\recomp(\gtext{})$ denote the RLSLP of $\gtext{}$ produced by $\recomp$.
We use the term \emph{letters} for variables introduced by $\recomp$.
Also, we use $c$ (rather than $X$) to represent a letter.

$\recomp$ consists of two different types of compression $\bcomp$ and $\pcomp$, which stand for Block Compression and Pair Compression, respectively.
\begin{itemize}
\item $\bcomp$: Given a string $w$ over $\Sigma = [1..|w|]$, 
      $\bcomp$ compresses $w$ by replacing all blocks of length $\geq 2$ with fresh letters.
      Note that $\bcomp$ eliminates all blocks of length $\geq 2$ in $w$.
      We can conduct $\bcomp$ in $O(|w|)$ time and space (see Lemma~\ref{lem:compute_bcomp}).
\item $\pcomp$: Given an string $w$ over $\Sigma = [1..|w|]$ that contains no block of length $\geq 2$, 
      $\pcomp$ compresses $w$ by replacing all pairs from $\lletters{} \rletters{}$ with fresh letters,
      where $(\lletters{}, \rletters{})$ is a partition of $\Sigma$, i.e., $\Sigma = \lletters{} \cup \rletters{}$ and $\lletters{} \cap \rletters{} = \emptyset$.
      We can deterministically compute in $O(|w|)$ time and space a partition of $\Sigma$ 
      by which at least $(|w|-1)/4$ pairs are replaced (see Lemma~\ref{lem:compute_partition}),
      and conduct $\pcomp$ in $O(|w|)$ time and space (see Lemma~\ref{lem:compute_pcomp}).
\end{itemize}
Let $\gtext{0}$ be a sequence of letters obtained by replacing every character $c$ of $\gtext{}$ with a letter generating $c$.
Then $\recomp$ compresses $\gtext{0}$ by applying $\bcomp$ and $\pcomp$ by turns until the string gets shrunk into a single letter.
Since $\pcomp$ compresses a given string by a constant factor $3/4$, the height of $\recomp(\gtext{})$ is $O(\lg N)$, 
and the total running time can be bounded by $O(N)$ (see Lemma~\ref{lem:compute_recomp}).

In order to give a formal description we introduce some notation below.
$\recomp$ transforms level by level $\gtext{0}$ into strings, $\gtext{1}, \gtext{2}, \dots, \gtext{\hat{h}}$
where $|\gtext{\hat{h}}| = 1$.
For any $0 \leq h \leq \hat{h}$, we say that $h$ is the \emph{level} of $\gtext{h}$.
If $h$ is even, the transformation from $\gtext{h}$ to $\gtext{h+1}$ is performed by $\bcomp$,
and production rules of the form $c \rightarrow \bchar{c}^d$ are introduced.
If $h$ is odd, the transformation from $\gtext{h}$ to $\gtext{h+1}$ is performed by $\pcomp$,
and production rules of the form $c \rightarrow \lchar{c} \rchar{c}$ are introduced.
Let $\letters{h}$ be the set of letters appearing in $\gtext{h}$.
For any even $h~(0 \leq h < \hat{h})$, let $\bletters{h}$ denote the set of letters with which there is a block of length $\geq 2$ in $\gtext{h}$.
For any odd $h~(0 \leq h < \hat{h})$, let $(\lletters{h}, \rletters{h})$ denote the partition of $\letters{h}$ used in $\pcomp$ of level $h$.

Figure~\ref{fig:ttog} shows an example of how $\recomp$ compresses $\gtext{0}$.
\begin{figure}[t]
\begin{center}
  \includegraphics[scale=0.5]{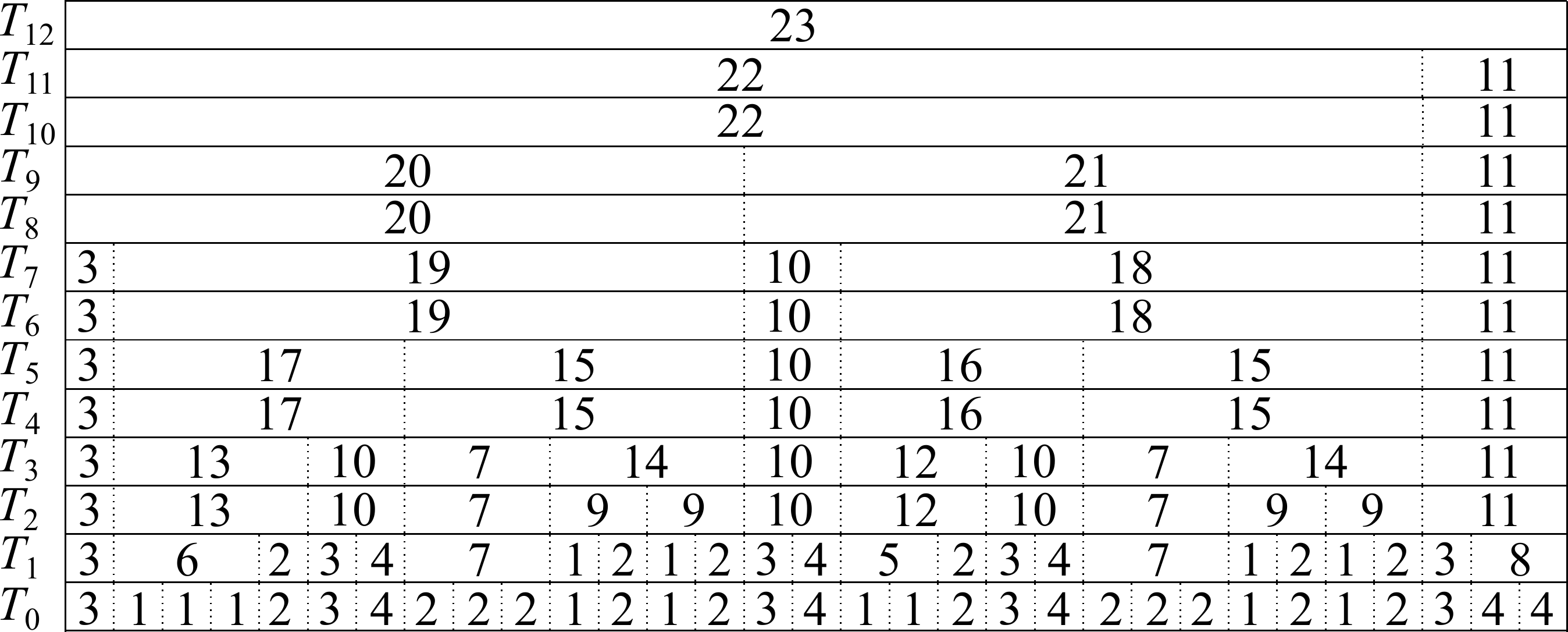}
  \caption{
  An example of how $\recomp$ compresses $\gtext{0}$.
  Below we enumerate non-empty $\bletters{h}, \lletters{h}, \rletters{h}$ and introduced production rules in each level.
  From $\gtext{0}$ to $\gtext{1}$: $\bletters{0} = \{ 1, 2, 4\}$, $\{5 \rightarrow 1^2, 6 \rightarrow 1^3, 7 \rightarrow 2^3, 8 \rightarrow 4^2\}$.
  From $\gtext{1}$ to $\gtext{2}$: $\lletters{1} = \{ 1, 3, 5, 6, 7\}$, $\rletters{1} = \{ 2, 4, 8\}$, 
  $\{9 \rightarrow (1, 2), 10 \rightarrow (3, 4), 11 \rightarrow (3, 8), 12 \rightarrow (5, 2), 13 \rightarrow (6, 2)\}$.
  From $\gtext{2}$ to $\gtext{3}$: $\bletters{2} = \{ 9\}$, $\{14 \rightarrow 9^2\}$.
  From $\gtext{3}$ to $\gtext{4}$: $\lletters{3} = \{ 3, 7, 12, 13\}$, $\rletters{3} = \{ 10, 14\}$, 
  $\{15 \rightarrow (7, 14), 16 \rightarrow (12, 10), 17 \rightarrow (13, 10)\}$.
  From $\gtext{5}$ to $\gtext{6}$: $\lletters{5} = \{ 3, 10, 11, 16, 17\}$, $\rletters{5} = \{ 15\}$, 
  $\{18 \rightarrow (16, 15), 19 \rightarrow (17, 15)\}$.
  From $\gtext{7}$ to $\gtext{8}$: $\lletters{7} = \{ 3, 10, 11\}$, $\rletters{7} = \{ 18, 19\}$, 
  $\{20 \rightarrow (3, 19), 21 \rightarrow (10, 18)\}$.
  From $\gtext{9}$ to $\gtext{10}$: $\lletters{9} = \{ 11, 20\}$, $\rletters{9} = \{ 21\}$, 
  $\{22 \rightarrow (20, 21)\}$.
  From $\gtext{11}$ to $\gtext{12}$: $\lletters{11} = \{ 22\}$, $\rletters{11} = \{ 11\}$, 
  $\{23 \rightarrow (22, 11)\}$.
  }
  \label{fig:ttog}
\end{center}
\end{figure}

The following four lemmas show how to conduct $\bcomp$, $\pcomp$, and therefore $\recomp$, efficiently,
which are essentially the same as respectively Lemma~2, Lemma~5, Lemma~6, and Theorem~1, stated in~\cite{Jez2015Aog}.
We give the proofs for the sake of completeness.

\begin{lemma}\label{lem:compute_bcomp}
  Given a string $w$ over $\Sigma = [1..|w|]$, 
  we can conduct $\bcomp$ in $O(|w|)$ time and space.
\end{lemma}
\begin{proof}
  We first scan $w$ in $O(|w|)$ time and list all the blocks of length $\geq 2$.
  Each block $c^d~(c \in \Sigma, d \geq 2)$ at position $i$ is listed by a triple $(c, d, i)$ of integers in $\Sigma$.
  Next we sort the list according to the pair of integers $(c, d)$, which can be done in $O(|w|)$ time and space by radix sort.
  Finally, we replace each block $c^d$ by a fresh letter based on the rank of $(c, d)$.
\end{proof}

For any string $w \in \Sigma^*$ that contains no block of length $\geq 2$,
let $\freq_{w}(c, \tilde{c}, 0)$ (resp. $\freq_{w}(c, \tilde{c}, 1)$) with $c > \tilde{c} \in \Sigma$ denote 
the number of occurrences of $c \tilde{c}$ (resp. $\tilde{c} c$) in $w$.
We refer to the list of non-zero $\freq_{w}(c, \tilde{c}, \cdot)$ sorted in increasing order of $c$ as the \emph{adjacency list} of $w$.
Note that it is the representation of the weighted directed graph in which 
there are exactly $\freq_{w}(c, \tilde{c}, 0)$ (resp. $\freq_{w}(c, \tilde{c}, 1)$) edges from $c$ to $\tilde{c}$ (resp.\ from $\tilde{c}$ to $c$).
Each occurrence of a pair in $w$ is counted exactly once in the adjacency list.
Then the problem of computing a good partition $(\lletters{}, \rletters{})$ of $\Sigma$ reduces to maximum directed cut problem on the graph.
Algorithm~\ref{fig:compute_partition} is based on a simple greedy $1/4$-approximation algorithm of maximum directed cut problem.

\begin{lemma}\label{lem:compute_partition}
  Given the adjacency list of size $m$ of a string $w \in \Sigma^*$,
  Algorithm~\ref{fig:compute_partition} computes in $O(m)$ time a partition $(\lletters{}, \rletters{})$ of $\Sigma$ 
  such that the number of occurrences of pairs from $\lletters{} \rletters{}$ in $w$ is at least $(|w|-1)/4$.
\end{lemma}
\begin{proof}
  In the foreach loop,
  we first run a $1/2$-approximation algorithm of maximum ``undirected'' cut problem on the adjacency list,
  i.e., we ignore the direction of the edges here.
  For each $c$ in increasing order, we greedily determine whether $c$ is added to $\lletters{}$ or to $\rletters{}$ depending on 
  $\sum_{\tilde{c} \in \rletters{}}\freq(c, \tilde{c}, \cdot) \geq \sum_{\tilde{c} \in \lletters{}}\freq(c, \tilde{c}, \cdot)$.
  Note that $\sum_{\tilde{c} \in \rletters{}}\freq(c, \tilde{c})$ (resp. $\sum_{\tilde{c} \in \lletters{}}\freq(c, \tilde{c})$) represents 
  the number of edges between $c$ and a character in $\rletters{}$ (resp. $\lletters{}$).
  By greedy choice, at least half of the edges in question become the ones connecting two characters each from $\lletters{}$ and $\rletters{}$.
  Hence, in the end, $|E|$ becomes at least $(|w|-1)/2$, 
  where let $E$ denote the set of edges between characters from $\lletters{}$ and $\rletters{}$
  (recalling that there are exactly $|w|-1$ edges).
  Since each edge in $E$ corresponds to an occurrence of a pair from $\lletters{} \rletters{} \cup \rletters{} \lletters{}$ in $w$,
  at least one of the two partitions $(\lletters{}, \rletters{})$ and $(\rletters{}, \lletters{})$ covers more than half of $E$.
  Hence we achieve our final bound $|E|/2 = (|w|-1)/4$ by choosing an appropriate partition at Line~\ref{line:switch}.

  In order to see that Algorithm~\ref{fig:compute_partition} runs in $O(m)$ time, 
  we only have to care about Line~\ref{line:compare_freq} and Line~\ref{line:switch}.
  We can compute $\sum_{\tilde{c} \in \rletters{}}\freq(c, \tilde{c}, \cdot)$ and $\sum_{\tilde{c} \in \lletters{}}\freq(c, \tilde{c}, \cdot)$ 
  by going through all $\freq(c, \cdot, \cdot)$ for fixed $c$ in the adjacency list, which are consecutive in the sorted list.
  Since each element of the list is used only once, the cost for Line~\ref{line:compare_freq} is $O(m)$ in total.
  Similarly the computation at Line~\ref{line:switch} can be done by going through the adjacency list again.
  Thus the algorithm runs in $O(m)$ time.
\end{proof}

\begin{algorithm2e}[!ht]
\caption{How to compute a partition of $\Sigma$ for $\pcomp$ to compress $w$ by $3/4$.}
\label{fig:compute_partition}
\SetKw{True}{true}
\SetKw{Or}{or}
\SetKw{Output}{output}
\SetKw{Return}{return}
\KwIn{Adjacency list of $w \in \Sigma^*$.}
\KwOut{$(\lletters{}, \rletters{})$ s.t. \# occurrences of pairs from $\lletters{} \rletters{}$ in $w$ is at least $(|w|-1)/4$.}
\tcc{The information whether $c \in \Sigma$ is in $\lletters{}$ or $\rletters{}$ is written in the data space for $c$, which can be accessed in $O(1)$ time.}
$\lletters{} \leftarrow \rletters{} \leftarrow \emptyset$\;
\ForEach{$c \in \Sigma$ in increasing order}{\label{line:firstpart_begin}
  \If{$\sum_{\tilde{c} \in \rletters{}}\freq_{w}(c, \tilde{c}, \cdot) \geq \sum_{\tilde{c} \in \lletters{}}\freq_{w}(c, \tilde{c}, \cdot)$}{\label{line:compare_freq}
    add $c$ to $\lletters{}$\;
  }\Else{
    add $c$ to $\rletters{}$\;\label{line:firstpart_end}
  }
}
\If{\# occurrences of pairs from $\lletters{}\rletters{} <$ \# occurrences of pairs from $\rletters{}\lletters{}$}{\label{line:switch}
  switch $\lletters{}$ and $\rletters{}$\;
}
\Return $(\lletters{}, \rletters{})$\;
\end{algorithm2e}

\begin{lemma}\label{lem:compute_pcomp}
  Given a string $w$ over $\Sigma = [1..|w|]$ that contains no block of length $\geq 2$, 
  we can conduct $\pcomp$ in $O(|w|)$ time and space.
\end{lemma}
\begin{proof}
  We first compute the adjacency list of $w$.
  This can be easily done in $O(|w|)$ time and space by sorting the $|w|-1$ size multiset 
  $\{(w[i], w[i+1], 0) \mid 1 \leq i < |w|, w[i] > w[i+1]\} \cup \{(w[i], w[i+1], 1) \mid 1 \leq i < |w|, w[i] < w[i+1]\}$ by radix sort.
  Then by Lemma~\ref{lem:compute_partition} 
  we compute a partition $(\lletters{}, \rletters{})$ in linear time in the size of the adjacency list, which is $O(|w|)$.
  Next we scan $w$ in $O(|w|)$ time and list all the occurrences of pairs to be compressed.
  Each pair $\lchar{c} \rchar{c} \in \lletters{} \rletters{}$ at position $i$ is listed by a triple $(\lchar{c}, \rchar{c}, i)$ of integers in $\Sigma$.
  Then we sort the list according to the pair of integers $(\lchar{c}, \rchar{c})$, which can be done in $O(|w|)$ time and space by radix sort.
  Finally, we replace each pair with a fresh letter based on the rank of $(\lchar{c}, \rchar{c})$.
\end{proof}

\begin{lemma}\label{lem:compute_recomp}
  Given a string $\gtext{}$ over $\Sigma = [1..N^{O(1)}]$, we can compute $\recomp(\gtext{})$ in $O(N)$ time and space.
\end{lemma}
\begin{proof}
  We first compute $\gtext{0}$ in $O(N)$ by sorting the characters used in $\gtext{}$ 
  and replacing them with ranks of characters.
  Then we compress $\gtext{0}$ by applying $\bcomp$ and $\pcomp$ by turns and get $\gtext{1}, \gtext{2} \dots \gtext{\hat{h}}$.
  One technical problem is that characters used in an input string $w$ of $\bcomp$ and $\pcomp$ should be in $[1..|w|]$, 
  which is crucial to conduct radix sort efficiently in $O(|w|)$ time~(see Lemmas~\ref{lem:compute_bcomp} and~\ref{lem:compute_pcomp}).
  However letters in $\gtext{h}$ do not necessarily hold this property.
  To overcome this problem, during computation we maintain ranks of letters among those used in the current $\gtext{h}$, which should be in $[1..|\gtext{h}|]$,
  and use the ranks instead of letters for radix sort.
  If we have such ranks in each level, we can easily maintain them by radix sort for the next level.
  Now, in every level $h~(0 \leq h < \hat{h})$ the compression from $\gtext{h}$ to $\gtext{h+1}$ can be conducted in $O(|\gtext{h}|)$ time and space.
  Since $\pcomp$ compresses a given string by a constant factor, the total running time can be bounded by $O(N)$ time.
\end{proof}

\subsection{Popped sequences}\label{sec:popseq}

We give a new characterization of recompression, 
which is a key to fast LCE queries as well as the upper bound $O(z \lg (N/z))$ for the size of $\recomp(\gtext{})$.
For any substring $w$ of $\gtext{}$, we define the \emph{Popped Sequence (PSeq)} of $w$, denoted by $\PSeq(w)$.
$\PSeq(w)$ is a sequence of letters such that $\val{}(\PSeq(w)) = w$ and consists of $O(\lg N)$ blocks of letters.
It is not surprising that any substring can be represented by $O(\lg N)$ blocks of letters
because the height of $\recomp(\gtext{})$ is $O(\lg N)$.
The significant property of $\PSeq(w)$ is that it occurs at ``every'' occurrence of $w$.
A similar property has been observed in CFGs produced by locally consistent parsing
and utilized for compressed indexes~\cite{Maruyama2013ESPIndex,Nishimoto2016DynamicIndexAndLZ} and 
a dynamic compressed LCE data structure~\cite{Nishimoto2016DynamicLCE_CompressedSpace}.
For example, in~\cite{Nishimoto2016DynamicIndexAndLZ,Nishimoto2016DynamicLCE_CompressedSpace} 
the sequence having such a property is called the \emph{common sequence} of $w$
but its representation size is $O(\lg |w| \lg^* N)$ rather than $O(\lg N)$.

$\PSeq(w)$ is the sequence of letters characterized by the following procedure.
Let $w_{0}$ be the substring of $\gtext{0}$ that generates $w$.
We consider applying $\bcomp$ and $\pcomp$ to $w_{0}$ exactly as we did to $\gtext{}$
but in each level we \emph{pop} some letters out if the letters can be coupled with letters outside the scope.
Formally, in increasing order of $h \geq 0$, we get $w_{h+1}$ from $w_{h}$ as follows:
\begin{itemize}
\item If $h$ is even. We first pop out the leftmost and rightmost blocks of $w_{h}$ if they are blocks of letter $c \in \bletters{h}$.
      Then we get $w_{h+1}$ by applying $\bcomp$ to the remaining string.
\item If $h$ is odd. We first pop out the leftmost letter and rightmost letter of $w_{h}$ if they are letters in $\rletters{h}$ and $\lletters{h}$, respectively.
      Then we get $w_{h+1}$ by applying $\pcomp$ to the remaining string.
\end{itemize}
We iterate this until the string disappears.
$\PSeq(w)$ is the sequence obtained by concatenating the popped-out letters/blocks in an appropriate order.
Note that for any occurrence of $w$ the letters inside the $\PSeq(w)$ are compressed in the same way.
Hence $w_{h}$ is created for every occurrence of $w$ and the occurrence of $\PSeq(w)$ is guaranteed (see also Figure~\ref{fig:pseq}).

\begin{figure}[t]
\begin{center}
  \includegraphics[scale=0.5]{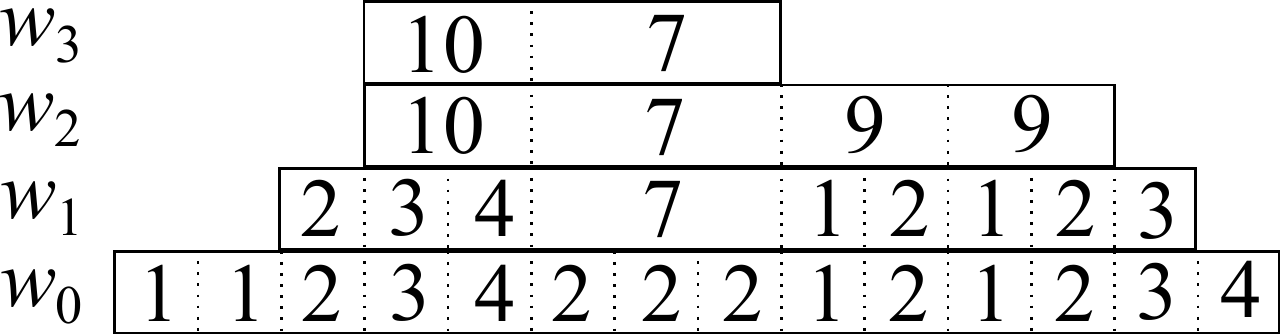}
  \caption{
  $\PSeq$ for $w_0 = (1, 1, 2, 3, 4, 2, 2, 2, 1, 2, 1, 2, 3, 4)$ under $\bletters{h}, \lletters{h}, \rletters{h}$ of Figure~\ref{fig:ttog}.
  At level $0$, a block of $1$ (resp. $4$) is popped out from the leftend (resp. rightend) of $w_0$ because $1, 4 \in \bletters{0}$.
  At level $1$, a letter $2$ (resp. $3$) is popped out from the leftend (resp. rightend) of $w_1$ because $2 \in \rletters{1}$ and $3 \in \lletters{1}$.
  At level $2$, a block of $9$ is popped out from the rightend of $w_2$ becaue $9 \in \bletters{2}$.
  At level $3$, a letter $10$ (resp. $7$) is popped out from the leftend (resp. rightend) of $w_3$ because $10 \in \rletters{1}$ and $7 \in \lletters{1}$.
  Then, $\PSeq(w_0) = (1, 1, 2, 10, 7, 9, 9, 3, 4)$.
  Observe that $w_0$ occurs twice in $\gtext{0}$ of Figure~\ref{fig:ttog}.
  and $w_0, w_1, w_2$ and $w_3$ are created over both occurrences. As a result, $\PSeq(w_0)$ occurs everywhere $w_{0}$ occurs.
  }
  \label{fig:pseq}
\end{center}
\end{figure}

The next lemma formalizes the above discussion.
\begin{lemma}\label{lem:popped_seq}
  For any substring $w$ of $\gtext{}$, $\PSeq(w)$ consists of $O(\lg N)$ blocks of letters.
  In addition, $w$ occurs at position $i$ iff $\PSeq(w)$ occurs at $i$.
\end{lemma}

\begin{lemma}\label{lem:num_ancestors}
  For any chain $C$ whose label consists of $m$ blocks of letters, the number of ancestor nodes of $C$ is $O(m)$.
\end{lemma}
\begin{proof}
  Since a block is compressed into one letter, the number of parent nodes of $C$ is at most $m$.
  As every internal node has two or more children, it is easy to see that there are $O(m)$ ancestor nodes of the parent nodes of $C$.
\end{proof}

\begin{corollary}\label{cor:num_ancestors_pseq}
  For any chain $C$ corresponding to $\PSeq(\gtext{}[b..e])$ for some interval $[b..e]$, the number of ancestor nodes of $C$ is $O(\lg N)$.
\end{corollary}

\begin{lemma}\label{lem:recomp_size}
  The size of $\recomp(\gtext{})$ is $O(z \lg (N/z))$, where $z$ is the size of the LZ77 factorization of $\gtext{}$.
\end{lemma}
\begin{proof}
  We first show the bound $O(z \lg N)$ and improve the analysis to $O(z \lg (N/z))$ later.

  Let $f_1 \dots f_z$ be the LZ77 factorization of $\gtext{}$.
  For any $1 \leq i \leq z$, let $L_{i}$ be the set of letters used in the ancestor nodes of leaves corresponding to $f_1 f_2 \dots f_i$.
  Clearly $|L_{1}| = O(\lg N)$. For any $1 < i \leq z$, we estimate $|L_{i} \setminus L_{i-1}|$.
  Since $f_i$ occurs in $f_1 \dots f_{i-1}$, we can see that the letters of $\PSeq(f_i)$ are in $L_{i-1}$ thanks to Lemma~\ref{lem:popped_seq}.
  Let $C_i$ be the chain corresponding to the occurrence $|f_1 \dots f_{i-1} + 1|$ of $\PSeq(f_i)$.
  Then, the letters in $L_{i} \setminus L_{i-1}$ are only in the labels of ancestor nodes of $C_i$.
  Since $\PSeq(f_i)$ consists of $O(\lg N)$ blocks of letters,
  $|L_{i} \setminus L_{i-1}|$ is bounded by $O(\lg N)$ due to Lemma~\ref{lem:num_ancestors}.
  Therefore the size of $\recomp(\gtext{})$ is $O(z \lg N)$.

  In order to improve the bound to $O(z \lg (N/z))$, we employ the same trick that has been used in the literature.
  Let $h = 2 \lg_{4/3} (N/z) = 2 \lg_{3/4} (z/N)$. Recall that $\pcomp$ compresses a given string by a constant factor $3/4$.
  Since $\pcomp$ has been applied $h/2$ times until the level $h$, $|\gtext{h}| \leq N (3/4)^{h/2} = z$, and hence, 
  the number of letters introduced in level $\geq h$ is bounded by $O(z)$.
  Then, we can ignore all the letters introduced in level $\geq h$ in the analysis of the previous paragraph, 
  and by doing so, the bound $O(\lg N)$ of $|L_{i} \setminus L_{i-1}|$ is improved to $O(h) = O(\lg (N/z))$.
  This yields the bound $O(z \lg (N/z))$ for the size of $\recomp(\gtext{})$.
\end{proof}

\begin{lemma}\label{lem:LCE_query}
  Given $\recomp(\gtext{})$, we can answer $\LCE(i,j)$ in $O(\lg N)$ time.
\end{lemma}
\begin{proof}
  Let $w$ be the longest common prefix of two suffixes beginning at $i$ and $j$.
  In the light of Lemma~\ref{lem:popped_seq}, $\PSeq(w)$, which consists of $O(\lg N)$ blocks of letters, occurs at both $i$ and $j$.
  Let $C_i$ (resp. $C_j$) be the chain that is labeled with $\PSeq(w)$ and begins at $i$ (resp. $j$).
  We can compute $|w|$ by traversing the ancestor nodes of $C_i$ and $C_j$ simultaneously and 
  matching $\PSeq(w)$ written in the labels of $C_i$ and $C_j$.
  Note that we do matching from left to right and we do not have to know $|w|$ in advance.
  Also, matching a block of letters in $\PSeq(w)$ can be done in $O(1)$ time on run-length encoded rules.
  By Corollary~\ref{cor:num_ancestors_pseq}, the number of ancestor nodes we have to visit is bounded by $O(\lg N)$.
  Thus, we get the lemma.
\end{proof}

\subsection{Proof of Theorem~\ref{theo:lce_from_uncompressed}}

\begin{proof}[Proof of Theorem~\ref{theo:lce_from_uncompressed}]
  By Lemma~\ref{lem:compute_recomp} we can compute $\recomp(\gtext{})$ in $O(N)$ time and space.
  Since the height of $\recomp(\gtext{})$ is $O(\lg N)$, we can support $\extract(i, \ell)$ queries in $O(\lg N + \ell)$ time due to Lemma~\ref{lem:extract}.
  $\LCE$ queries can be supported in $O(\lg N)$ time by Lemma~\ref{lem:LCE_query}.
\end{proof}

\section{LCE data structure built from SLPs}

In this section, we prove Theorem~\ref{theo:lce_from_slp}.
Input is now an arbitrary SLP $\slp = \{\Sigma, \vars, \rules\}$ of size $n$ generating $\gtext{}$.
Basically what we consider is to simulate $\recomp$ on $\slp$, namely, compute $\recomp(\gtext{})$ without decompressing $\slp$ explicitly.
The recompression technique is celebrated for doing this kind of tasks (actually this is where ``recompression'' is named after).
In Section~\ref{sec:simulate_recomp}, we present an algorithm $\simrecomp$ that simulates $\recomp$ in $O(n \lg^2 (N/n))$ time and $O(n + z \lg (N/z))$ space.
In Section~\ref{sec:gtog}, 
we present how to modify $\simrecomp$ to obtain an $O(n \lg (N/n))$-time and $O(n + z \lg (N/z))$-space construction of our LCE data structure.

\subsection{$\simrecomp$: Simulating $\recomp$ on CFGs}\label{sec:simulate_recomp}

We present an algorithm $\simrecomp$ to simulate $\recomp$ on $\slp$.
To begin with, we compute the CFG $\rg{0} = \{ \letters{0}, \vars, \hrules{0}\}$ obtained by 
replacing, for all variables $X \in \vars$ with $\rules(X) \in \Sigma$, 
every occurrence of $X$ in the righthand sides of $\rules$ with the letter generating $\rules(X)$.
Note that $\letters{0}$ is the set of terminals of $\rg{0}$, and $\rg{0}$ generates $\gtext{0}$.
$\simrecomp$ transforms level by level $\rg{0}$ into CFGs, 
$\rg{1} = \{ \letters{1}, \vars, \hrules{1}\}, \rg{2} = \{ \letters{2}, \vars, \hrules{2}\}, \dots, \rg{\hat{h}} = \{ \letters{\hat{h}}, \vars, \hrules{\hat{h}}\}$,
where each $\rg{h}$ generates $\gtext{h}$.
Namely, compression from $\gtext{h}$ to $\gtext{h+1}$ is simulated on $\rg{h}$.
We can correctly compute the letters in $\nletters{h+1}$ while modifying $\rg{h}$ into $\rg{h+1}$,
and hence, we get all the letters of $\recomp(\gtext{})$ in the end.
We note that new variables are never introduced and 
the modification is done by rewriting righthand sides of the original variables.

Here we introduce the special formation of the CFGs $\rg{h}$ (it is a generalization of SLPs in a different sense from RLSLPs):
For any $X \in \vars$, $\hrules{h}(X)$ consists of an ``arbitrary number'' of letters and at most ``two'' variables.
More precisely, the following condition holds:
\begin{list}{}{}
\item For any variable $X \in \vars$ with $\rules(X) = \lchar{X}\rchar{X}$, 
      $\hrules{h}(X)$ is either $w_1 \lchar{X} w_2 \rchar{X} w_3$, $w_1 \lchar{X} w_2$, $w_2 \rchar{X} w_3$ or $w_2$ 
      with $w_1, w_2, w_3 \in \letters{h}^*$, where $w_1 = w_3 = \emptystr$ if $X$ is not the starting variable.
\end{list}
As opposed to SLPs and RLSLPs, we define the size of $\rg{h}$ by the total lengths of righthand sides and denote it by $|\rg{h}|$.

\subsubsection{$\pcomp$ on CFGs}
We firstly show that the adjacency list of $\gtext{h}$ can be computed efficiently.
\begin{lemma}[Lemma~6.1 of~\cite{Jez2015FFC}]\label{lem:adjacency_list_cfg}
  For any odd $h~(0 \leq h < \hat{h})$, the adjacency list of $\gtext{h}$, whose size is $O(|\rg{h}|)$, 
  can be computed in $O(|\rg{h}|+n)$ time and space.
\end{lemma}
\begin{proof}
  For any variable $X \in \vars$, let $\vocc{X}$ denote the number of occurrences of the nodes labeled with $X$ in the derivation tree of $\slp$.
  It is well known that $\vocc{X}$ for all variables can be computed in $O(n)$ time and space 
  on the DAG representation of the tree.\footnote{It is enough to compute $\vocc{X}$ once at the very beginning of $\simrecomp$.}
  Also, for any variable $X \in \vars$, let $\lml{X}$ and $\rml{X}$ denote the leftmost letter and respectively rightmost letter of $\val{\rg{h}}(X)$.
  We can compute $\lml{X}$ for all variables in $O(|\rg{h}|)$ time by a bottom up computation, i.e., 
  $\lml{X} = \lml{Y}$ if $\hrules{h}(X)$ starts with a variable $Y$, and $\lml{X} = w[1]$ if $\hrules{h}(X)$ starts with a non-empty string $w$.
  In a completely symmetric way $\rml{X}$ can be computed in $O(|\rg{h}|)$ time.

  Now observe that any occurrence $i$ of a pair $\lchar{c} \rchar{c}$ in $\gtext{h}$ can be uniquely associated with 
  a variable $X$ that is the label of the lowest node covering the interval $[i..i+1]$ in the derivation tree of $\rg{h}$ (recall that $\rg{h}$ generates $\gtext{h}$).
  We intend to count all the occurrences of pairs associated with $X$ in $\hrules{h}(X)$.
  For example, let $\hrules{h}(X) = \lchar{X} w_2 \rchar{X}$ with $w_2 \in \letters{h}^*$.
  Then $\lchar{c}\rchar{c}$ appears \emph{explicitly} in $w_2$ or \emph{crosses} the boundaries of $\lchar{X}$ and/or $\rchar{X}$.
  If $\lchar{c}\rchar{c}$ crosses the boundary of $\lchar{X}$, $\rml{\lchar{X}}$ is $\lchar{c}$ and $\rchar{c}$ follows,
  i.e., $(w_2[1] = \rchar{c}) \vee (w_2 = \emptystr \wedge \lml{\rchar{X}} = \rchar{c})$.
  Using $\rml{\lchar{X}}$ and $\lml{\rchar{X}}$, we can compute in $O(|\hrules{h}(X)|)$ time and space
  a $(|\hrules{h}(X)| - 1)$-size multiset that lists all the explicit and crossing pairs in $\hrules{h}(X)$.
  Each pair $\lchar{c} \rchar{c}$ with $\lchar{c} > \rchar{c}$ (resp. $\lchar{c} < \rchar{c}$) 
  is listed by a quadruple $(\lchar{c}, \rchar{c}, 0, \vocc{X} \})$ (resp. $(\rchar{c}, \lchar{c}, 1, \vocc{X} \})$.
  $\vocc{X}$ means that the pair has a weight $\vocc{X}$ because the pair appears every time a node labeled with $X$ appears in the derivation tree.

  We compute such a multiset for every variable, which takes $O(|\rg{h}|)$ time and space in total.
  Next we sort the obtained list in increasing order of the first three integers in a quadruple.
  Note that the maximum value of letters is $O(z\lg (N/z))$ due to Lemma~\ref{lem:recomp_size},
  and $O(z \lg (N/z)) = O(n^2)$ since $z \leq n$ and $\lg N \leq n$ hold.
  Thus the sorting can be done in $O(n)$ time and space by radix sort.
  Finally we can get the adjacency list of $\gtext{h}$ by summing up weights of the same pair.
  The size of the list is clearly $O(|\rg{h}|)$.
\end{proof}

The next lemma shows how to implement $\pcomp$ on CFGs:
\begin{lemma}\label{lem:sim_pcomp}
  For any odd $h~(0 \leq h < \hat{h})$, we can compute $\rg{h+1}$ from $\rg{h}$ in $O(|\rg{h}|+n)$ time and space.
  In addition, $|\rg{h+1}| \leq |\rg{h}| + 2n$.
\end{lemma}
\begin{proof}
  We first compute the partition $(\lletters{h}, \rletters{h})$ of $\letters{h}$,
  which can be done in $O(|\rg{h}|+n)$ time and space by Lemmas~\ref{lem:adjacency_list_cfg} and~\ref{lem:compute_partition}.

  Given $(\lletters{h}, \rletters{h})$, 
  we can detect all the positions of the pairs from $\lletters{h} \rletters{h}$
  in the righthands of $\hrules{h}$, which should be compressed.
  Some of the appearances of the pairs are explicit and the others are crossing.
  While explicit pairs can be compressed easily, crossing pairs need an additional treatment.
  In order to deal with crossing pairs, 
  we first \emph{uncross} them 
  by popping out $\lml{Y}$ (resp. $\rml{Y}$) from $\val{\rg{h}}(Y)$ iff $\lml{Y} \in \rletters{h}$ (resp. $\rml{Y} \in \lletters{h}$) 
  for every variable $Y$ other than the starting variable.
  More precisely, we do the following:
  \begin{description}
  \item[$\pil$] For any variable $X$,
        if $\hrules{h}(X)[i] = Y \in \vars$ with $i > 1$ ($i \geq 1$ if $X$ is the starting variable) and $\lml{Y} \in \rletters{h}$, 
        replace the occurrence of $Y$ with $\lml{Y} Y$;
        if $\hrules{h}(X)[i] = Y \in \vars$ with $i < |\hrules{h}(X)|$ ($i \leq |\hrules{h}(X)|$ if $X$ is the starting variable) and $\rml{Y} \in \lletters{h}$, 
        replace the occurrence of $Y$ with $Y \rml{Y}$.
  \item[$\pol$] For any variable $X$ other than the starting variable,
        if $\hrules{h}(X)[1] \in \rletters{h}$, remove the first letter of $\hrules{h}(X)$; and
        if $\hrules{h}(X)[|\hrules{h}(X)|] \in \lletters{h}$, remove the last letter of $\hrules{h}(X)$.
        In addition, if $X$ becomes empty, we remove all the appearances of $X$ in $\hrules{h}$.
  \end{description}
  $\pol$ removes $\lml{Y}$ and $\rml{Y}$ from $\val{\rg{h}}(Y)$ if they can be a part of a crossing pair 
  and $\pil$ introduces the removed letters into appropriate positions in $\hrules{h}$ so that the modified $\rg{h}$ keeps to generate $\gtext{h}$.
  Notice that for each variable $X$ the positions where letters popped-in is at most two (four if $X$ is the starting variable)
  and there is at least one variable that has no variables below, and hence, the size of $\rg{h}$ increases at most $2n$.
  The uncrossing can be conducted in $O(|\rg{h}|+n)$ time.

  Since all the pairs to be compressed become explicit now, we can easily conduct $\bcomp$ in $O(|\rg{h}|+n)$ time.
  We scan righthand sides in $O(|\rg{h}|)$ time and list all the occurrences of pairs to be compressed.
  Each occurrence of pair $\lchar{c} \rchar{c} \in \lletters{} \rletters{}$ is listed 
  by a triple $(\lchar{c}, \rchar{c}, p)$, where $p$ is the pointer to the occurrence.
  Then we sort the list according to the pair of integers $(\lchar{c}, \rchar{c})$, 
  which can be done in $O(|\rg{h}|+n)$ time and space by radix sort because $\lchar{c}$ and $\rchar{c}$ are $O(n^2)$.
  Finally, we replace each pair at position $p$ with a fresh letter based on the rank of $(\lchar{c}, \rchar{c})$.
\end{proof}

\subsubsection{$\bcomp$ on CFGs}\label{sec:bcomp_on_CFG}
For any even $h~(0 \leq h < \hat{h})$, $\bcomp$ can be implemented in a similar way to $\pcomp$ of Lemma~\ref{lem:sim_pcomp}.
A block $\gtext{h}[b..e]$ of length $\geq 2$ is uniquely associated with a variable $X$ 
that is the label of the lowest node covering the interval $[b-1..e+1]$ in the derivation tree of $\rg{h}$ 
(if $b = 0$ or $e = |\gtext{h}|$, the block is associated with the starting variable).
Note that we take $[b-1..e+1]$ rather than $[b..e]$ to be sure that the block cannot extend outside the variable.
Some blocks are explicitly written in $\hrules{h}(X)$ and some others are crossing the boundaries of variables in $\hrules{h}(X)$.
The numbers of explicit blocks and crossing blocks in $\hrules{h}$ is at most $|\rg{h}|$ and $2n$, respectively.
The crossing blocks can be uncrossed in a similar way to uncrossing pairs.
Then $\bcomp$ can be done by replacing all the blocks with fresh letters on righthand sides of $\hrules{h}$.

However here we have a problem.
Recall that in order to give a unique letter to a block $c^d$, we have to sort the pairs of integers $(c, d)$ (see Lemma~\ref{lem:compute_bcomp}).
Since $d$ might be exponentially larger than $|\rg{h}|+n$, radix sort cannot be executed in $O(|\rg{h}|+n)$ time and space.
In Section~6.2 of~\cite{Jez2015FFC}, Je\.z showed how to solve this problem by tweaking the representation of lengths of long blocks, 
but its implementation and analysis are involved.\footnote{Note that Section~6.2 of~\cite{Jez2015FFC} also takes care of the case where the word size is $\Theta(\lg n)$ rather than $\Theta(\lg N)$. We do not consider the $\Theta(\lg n)$-bits model in this paper because using $\Theta(\lg N)$ bits to store the length of string generated by every letter is crucial for extract and LCE queries.\@ However, we believe that our new observation stated in Lemma~\ref{lem:short_block} will simplify the analysis for the $\Theta(\lg n)$-bits model, too.}

We show in Lemma~\ref{lem:short_block} our new observation, which leads to a simpler implementation and analysis of $\bcomp$.
We say that a block $c^d$ is \emph{short} if $d = O(|\rg{h}|+n)$ and \emph{long} otherwise.
Also, we say that a variable is \emph{unary} iff its righthand side consists of a single block.

\begin{lemma}\label{lem:short_block}
  For any even $h~(0 \leq h < \hat{h})$, a block $\gtext{h}[b..e] = c^d$ is short if it does not include a substring generated from a unary variable.
\end{lemma}
\begin{proof}
  Consider the derivation tree of $\rg{h}$ and the shortest path from $\gtext{h}[b]$ to $\gtext{h}[e]$.
  Let $X_1 X_2 \cdots X_{m'} \cdots X_m$ be the sequence of labels of internal nodes on the path, 
  where $X_{m'}$ corresponds to the lowest common ancestor of $\gtext{h}[b]$ and $\gtext{h}[e]$.
  Since SLPs have no loops in the derivation tree, $X_1, \dots, X_{m'}$ are all distinct.
  Similarly $X_{m'+1}, \dots, X_{m}$ are all distinct.
  Since a unary variable is not involved to generate the block, it is easy to see that $d \leq \sum_{i = 1}^{m} |\hrules{h}(X_i)| \leq 2|\rg{h}|$ holds.
\end{proof}

Lemma~\ref{lem:short_block} implies that most of blocks we find during the compression are short, which can be sorted efficiently by radix sort.
If there is a long block in $\hrules{h}$, an occurrence of a unary variable $X$ must be involved to generate the block.
Since $\bcomp$ at level $h$ pops out all the letters from $X$ and removes the occurrences of $X$ in $\hrules{h}$,
there are at most $2n$ long blocks in total.
The number of long blocks can also be upper bounded by $2N/n$ with a different analysis based on the following fact:
\begin{fact}\label{fact:overlap_long_block}
If a substring of original text $\gtext{}$ generated from a long block overlaps with that generated from another long block,
one substring must include the other, and moreover, the shorter block is completely included in ``one'' letter of the longer block.
Hence the length of the substring of the longer block is at least $n$ times longer than that of the shorter block.
\end{fact}
Let us consider the long blocks that generate substrings whose lengths are $[n^{i}..n^{i+1})$ for a fixed integer $i \geq 1$.
By Fact~\ref{fact:overlap_long_block}, the substrings cannot overlap, and hence, the number of such long blocks is at most $N/n^{i}$.
Therefore, the total number of long blocks is at most $\sum_{i \geq 1} N/n^{i} \leq 2N/n$.
Thus we get the following lemma.
\begin{lemma}\label{lem:ub_long_blocks}
  There are at most $O(\min(n, N / n))$ long blocks found during $\simrecomp$.
\end{lemma}

By Lemma~\ref{lem:ub_long_blocks}, we can employ a standard comparison-base sorting algorithm to sort all long blocks in $O(n \lg (\min(n, N/n)))$ time in total.
In particular, $\bcomp$ of one level can be implemented in the following complexities:
\begin{lemma}\label{lem:sim_bcomp}
  For any even $h~(0 \leq h < \hat{h})$, we can compute $\rg{h+1}$ from $\rg{h}$ in $O(|\rg{h}| + n + m \lg m))$ time and $O(|\rg{h}|+n)$ space,
  where $m$ is the number of long blocks in $\hrules{h}$.
  In addition, $|\rg{h+1}| \leq |\rg{h}| + 2n$.
\end{lemma}

\subsubsection{The complexities of $\simrecomp$}
\begin{theorem}\label{theo:simrecomp}
  $\simrecomp$ runs in $O(n \lg^2 (N/n))$ time and $O(n \lg (N/n))$ space.
\end{theorem}
\begin{proof}
  Using $\pcomp$ and $\bcomp$ implemented on CFGs (see Lemma~\ref{lem:sim_pcomp} and~\ref{lem:sim_bcomp}), 
  $\simrecomp$ transforms level by level $\rg{0}$ into $\rg{1}, \rg{2}, \dots, \rg{\hat{h}}$.
  In each level, the size of CFGs can increase at most $2n$ by the procedure of uncrossing.
  Since $|\rg{h}| = O(n \lg N)$ for any $h~(0 \leq h < \hat{h})$, 
  we get the time complexity $O(n \lg^2 N)$ by simply applying Lemmas~\ref{lem:sim_pcomp} and~\ref{lem:sim_bcomp}.

  We can improve it to $O(n \lg^2 (N/n))$ by a similar trick used in the proof of Lemma~\ref{lem:recomp_size}.
  At some level $h'$ where $|\gtext{h'}|$ becomes less than $n$, 
  we decompress $\rg{h'}$ and switch to $\recomp$, which transforms $\gtext{h'}$ into $\gtext{\hat{h}}$ in $O(n)$ time by Lemma~\ref{lem:compute_recomp}.
  We apply Lemmas~\ref{lem:sim_pcomp} and~\ref{lem:sim_bcomp} only for $h$ with $0 \leq h < h'$.
  Since $h' = O(\lg (N/n))$, $|\rg{h}| = O(n \lg (N/n))$ for any $h~(0 \leq h < h')$.
  Therefore, we get the time complexity $O(n \lg^2 (N/n))$.
  The space complexity is bounded by the maximum size of CFGs $\rg{0}, \rg{1}, \dots, \rg{h'}$, which is $O(n \lg (N/n))$.
\end{proof}

\subsection{$\gtog$: $O(n \lg (N/n))$-time recompression}\label{sec:gtog}
We modify $\simrecomp$ slightly to run in $O(n \lg (N/n))$ time and $O(n + z \lg (N/z))$ space.
The idea is the same as what has been presented in Section~6.1 of~\cite{Jez2015FFC}.
The problem of $\simrecomp$ is that the sizes of intermediate CFGs $\rg{h}$ can grow up to $O(n \lg (N/n))$.
If we can keep their sizes to $O(n)$, everything goes fine.
This can be achieved by using two different types of partitions of $\letters{h}$ for $\pcomp$:
One is for compressing $\gtext{h}$ by a constant factor, and the other for compressing $|\rg{h}|$ by a constant factor (unless $|\rg{h}|$ is too small to compress).
Recall that the former partition has been used in $\recomp$ and $\simrecomp$, 
and the partition is computed from the adjacency list of $\gtext{h}$ by Algorithm~\ref{fig:compute_partition}.
Algorithm~\ref{fig:compute_partition} can be extended to work on a set of strings 
by just inputting the adjacency list from a set of strings.
Then, we can compute the partition for compressing $|\rg{h}|$ by a constant factor 
by considering the adjacency list from a set of strings in the righthand sides of $\hrules{h}$.
The adjacency list can be easily computed in $O(|\rg{h}|+n)$ time and space 
by modifying the algorithm described in the proof of Lemma~\ref{lem:adjacency_list_cfg}:
We just ignore the weight $\vocc{X}$, i.e., use a unit weight $1$ for every listed pair.
Using the two types of partitions alternately, 
we can compress strings by a constant factor while keeping the sizes of the intermediate CFGs to $O(n)$.

We denote the modified algorithm by $\gtog$ and the resulting RLSLP by $\gtog(\slp)$.
Note that $\gtog(\slp)$ is not identical to $\recomp(\gtext{})$ in general 
because the partitions used in $\gtog$ change depending on the input $\slp$.
Still the height of $\gtog(\slp)$ is $O(\lg N)$ and the properties of $\PSeq$s hold.
Hence we can support LCE queries on $\gtog(\slp)$ as we did on $\recomp(\gtext{})$ by Lemma~\ref{lem:LCE_query}.

\subsection{Proof of Theorem~\ref{theo:lce_from_slp}}
\begin{proof}[Proof of Theorem~\ref{theo:lce_from_slp}]
  Let $\slp$ be an input SLP of size $n$ generating $\gtext{}$.
  We compute $\gtog(\slp)$ in $O(n \lg (N/n))$ time and $O(n + z \lg (N/z))$ space as described in Section~\ref{sec:gtog}.
  Since the height of $\recomp(\gtext{})$ is $O(\lg N)$, we can support $\extract(i, \ell)$ queries in $O(\lg N + \ell)$ time due to Lemma~\ref{lem:extract}.
  $\gtog(\slp)$ supports $\LCE$ queries in $O(\lg N)$ time in the same way as what was described in Lemma~\ref{lem:LCE_query}.
\end{proof}

\subparagraph*{Acknowledgements.}
The author was supported by JSPS KAKENHI Grant Number 16K16009.

\clearpage
\bibliography{ref}

\end{document}